\documentclass[twoside]{article}
\usepackage{aistats2014}
\usepackage[round]{natbib}
\usepackage{float}

\usepackage{url}
\usepackage{algorithm}
\usepackage{algorithmic}
\usepackage{amsmath, amssymb, amsfonts, amsthm}
\usepackage{amstext}
\newtheorem{defn}{Definition}
\newtheorem{thm}{Theorem}
\newtheorem{prop}{Proposition}
\usepackage{mathtools}
\DeclarePairedDelimiter{\diagpars}{(}{)}
\newcommand{\diag}{\operatorname{diag}\diagpars}

\usepackage{graphicx}
\usepackage{color}
\usepackage{bm}
\usepackage{subfigure}
\begin{document}

%

%

\twocolumn[

\aistatstitle{Scalable multiscale density estimation}

\aistatsauthor{ Ye Wang \And Antonio Canale \And David Dunson}

\aistatsaddress{
Duke University\\
\texttt{eric.ye.wang@duke.edu} \\
 \And 
Universit\` a degli studi di \\
Torino e Collegio Carlo Alberto\\
\texttt{antonio.canale@unito.it} \\
 \And 
Duke University\\
\texttt{dunson@stat.duke.edu}}
]

\begin{abstract}
Although Bayesian density estimation using discrete mixtures has good
performance in modest dimensions, there is a lack of statistical and
computational scalability to high-dimensional multivariate cases.
To combat the curse of dimensionality, it is necessary
to assume the data are concentrated near a lower-dimensional subspace.
However, Bayesian methods for learning this subspace along with the
density of the data scale poorly computationally. To solve this problem,
we propose an empirical Bayes approach, which estimates a multiscale
dictionary using geometric multiresolution analysis in a first stage.
We use this dictionary within a multiscale mixture model, which allows
uncertainty in component allocation, mixture weights and scaling factors
over a binary tree. A computational algorithm is proposed, which scales
efficiently to massive dimensional problems. We provide some theoretical
support for this geometric density estimation (GEODE) method, and illustrate the performance through simulated
and real data examples.
\end{abstract}

\section{Introduction}
\label{sec:Intro}

Let $\boldsymbol{y}_{i}=(y_{i1},\ldots,y_{iD})^{T}$, for $i=1,\ldots,n$, be a sample from an unknown distribution having support in a subset of  $\Re^{D}$. We are interested in estimating its density when $D$ is large, and the data have a low-dimensional structure with intrinsic dimension $p$ such that $p\ll D$. Kernel methods work well in low dimensions, but face challenges in scaling up to large $D$ settings.  In particular, optimally one would allow separate bandwidth parameters for the different variables to accommodate differing smoothness, but then there is the issue of how to choose the high-dimensional vector of bandwidths or alternatively the kernel covariance matrix.  Clearly, cross validation involves an intractable computation cost and plugging in arbitrary values is not recommended, since bandwidth choice fundamentally impacts performance \citep{liu2007sparse}. Bayesian nonparametric models \citep{escobar1995bayesian,rasmussen1999infinite} provide an alternative approach for density estimation, specifying priors for the bandwidth parameters allowing adaptive estimation without cross-validation \citep{shen2011adaptive}.  However, inference is prohibitively costly. To scale up nonparametric Bayes inference, one can potentially rely on maximum a posteriori  (MAP) estimation \citep{ghahramani1996algorithm} or variational Bayes (VB) \citep{ghahramani1999variational}.  Issues with MAP include difficulties in efficient estimation in high-dimensions, with the EM algorithm tending to converge slowly to a local mode, and lack of characterization of uncertainty. Although VB provides an approximation to the full posterior instead of just the mode, it is well known that posterior uncertainty is substantially underestimated \citep{wang2004inadequacy} and in being implemented with EM, VB inherits the computational problems of MAP estimation. 

Manifold learning methods \citep{ tenenbaum2000global, lawrence2005probabilistic} provide computationally efficient and geometric-oriented dimension reduction, motivating an alternative way to characterize the density via a low-dimensional embedding. While most of these methods have focused on visualization, manifold Parzen windows \citep{vincent2003manifold} is a notable exception that has attempted to combine density estimation and manifold learning. The model applies dimension reduction and fits a Gaussian ``pancake" to the neighbourhood area of each data point, integrating local geometric information into a kernel density estimator. However, overfitting might come in when every data point is associated, by the same weight, with a Gaussian. Moreover, the model can be sensitive to the prior choice of intrinsic dimension $p$, and only provides a point estimate. We addressed these problems by designing an empirical Bayes nonparametric density estimator based on a set of multiscale geometric dictionaries learned at a first stage. The proposed estimator combines density estimation and manifold learning, characterizes uncertainty, scales up to problems with massive dimensions and is capable of automatically learning the intrinsic dimension. The model is illustrated through simulated and real data examples.

The remainder of the paper is organized as follows. Our geometric density estimation (GEODE), consisting of first stage dictionary learning followed by rapid Bayesian inference, is proposed in \S~\ref{Sec:MDLF}. The performance of the proposed method is tested through simulation experiments in \S~\ref{Sec:SimulationStudies} and real data applications to image inpainting data handwritten digit classification data in \S~\ref{Sec:App}. A discussion is reported in \S~\ref{Sec:Discussion}.

\section{Bayes dictionary learning in factor models}
\label{Sec:MDLF}

Assume $\boldsymbol{y}_{i}\sim \mathcal{N}_{D}(\boldsymbol{\mu},\boldsymbol{\Omega})$, with $\boldsymbol{\mu}\in\Re^{D}$ a mean vector and $\boldsymbol{\Omega} \in \Re^{D\times D}$ a covariance matrix, for $i=1,2,\dots,n$. An efficient approach to reduce  dimension when $D$ is large relies on the factor analytic decomposition $\boldsymbol{\Omega}=\boldsymbol{\Lambda}\boldsymbol{\Lambda}^{T}+\sigma^{2}\boldsymbol{I}$, where $\boldsymbol{\Lambda}$ is a $D\times p$ matrix with $p\ll D$. \citet{carvalho2008high} and \citet{bhattacharya2011sparse} (among many others) have successfully applied FA under the Bayesian paradigm while additionally assuming $\boldsymbol{\Lambda}$  sparse. The mixture of factor analyzers (MFA) model extends FA to be able to characterize non-Gaussian data.  Bayesian MFA is straightforward to implement in small dimensional problems  \citep{diebolt1994estimation,richardson1997bayesian}, but faces problems in scaling beyond a few 100 dimensions.

To simplify computation, we propose an empirical Bayes approach that avoids directly placing priors on selected parameters in the factorizations via the use of multiscale dictionary learning.

\subsection{Formulation}
\label{ssec:Formulation}

The MFA model is given by
\begin{equation}
\label{eq:MFA}
f(\boldsymbol{y}_{i}) \sim \sum_{k=1}^{K} \pi_{k}\mathcal{N}_D\big(\boldsymbol{\mu}_{k},\boldsymbol{\Lambda}_{k}\boldsymbol{\Lambda}_{k}^{T}+\sigma^{2}_{k}\boldsymbol{I}\big),
\end{equation}
where $K$ is the number of components, $\boldsymbol{\mu}_{k}\in\Re^{D}$ is a mean vector and $\pi_{k}$ is the mixing weight for the $k$th component with $\sum_{k=1}^{K} \pi_{k}=1$. The intrinsic dimension $p$ is not observable;  we start with a guess $d$ with $\boldsymbol{\Lambda}_{k}$ a $D\times d$ matrix, for $k=1,\dots,K$. Later we will discuss how we can efficiently learn $p$. MFA assumes the data are centered around multiple low\textendash{}dimensional linear subspaces $\mbox{span}({\boldsymbol{\Lambda}_{k}})$, for $k=1,\dots,K$. Let $\boldsymbol{\Phi}_{k}$ be a $D\times d$ matrix with column vectors being the basis for  $\mbox{span}({\boldsymbol{\Lambda}_{k}})$. 

For simplicity, we assume the column vectors of $\boldsymbol{\Phi}_{k}$ and the column vectors of $\boldsymbol{\Lambda}_{k}$ are in same directions. Then the MFA model can be written as
\begin{equation}
\label{eq:MDicFA}
f(\boldsymbol{y}_{i}) \sim \sum_{k=1}^{K} \pi_{k}\mathcal{N}_D\big(\boldsymbol{\mu}_{k},\boldsymbol{\Phi}_{k}\boldsymbol{\Sigma}_{k}\boldsymbol{\Phi}_{k}^{T}+\sigma^{2}_{k}\boldsymbol{I}\big),
\end{equation}
where $\boldsymbol{\Sigma}_{k}$ is a $d\times d$ positive diagonal matrix, for $k=1,\dots,K$.

If $\boldsymbol{\mu}_{k}$ and $\boldsymbol{\Phi}_{k}$ are fixed, the Bayesian learning in high dimensions is clearly greatly simplified, since instead of $\boldsymbol{\Lambda}_{k}$ and $\boldsymbol{\mu}_{k}$, only $\boldsymbol{\Sigma}_{k}$ and $\sigma^2_{k}$, for $k=1,\dots,K$, needs to be learned. However, this modification inherits from MFA the problem of choosing $K$ and $d$, and relies heavily on the quality of the pre\textendash{}learned dictionaries. To address the problem, we propose a multiscale mixture generalization based on a set of pre\textendash{}learned multiscale dictionaries $\big\{\boldsymbol{\mu}_{s,h},\boldsymbol{\Phi}_{s,h}\big\}$ where $(s,h)$ denotes the node index of a binary clustering tree. The dictionaries are obtained in a first stage using geometric multi\textendash{}resolution analysis (GMRA) \citep{allard2012multi}, which is shown to be capable of providing high\textendash{}quality basis vectors for local linear subspaces at different scales. We call this method the geometric density estimation (GEODE), which can be written as
\begin{equation}
\label{eq:MDLF}
f(\boldsymbol{y}_{i}) \sim \sum_{s,h} \pi_{s,h}\mathcal{N}_D\big(\boldsymbol{\mu}_{s,h},\boldsymbol{\Phi}_{s,h}\boldsymbol{\Sigma}_{s,h}\boldsymbol{\Phi}_{s,h}^{T}+\sigma^{2}_{s}\boldsymbol{I}\big),
\end{equation}
where $\boldsymbol{\Sigma}_{s,h}=\diag{\alpha^{2}_{s,h,1},\dots,\alpha^{2}_{s,h,d}}$. The proposed method mixes flexibly across a binary tree, both across scales and within scales in a Bayesian manner and hence tends to better capture the nonlinear structure and be more resistant to over\textendash{}fitting. Moreover, the method is capable of adaptively removing redundant dimensions and efficiently learning the true intrinsic dimension $p$. Both aspects will be demonstrated in more details later.

Borrowing the notations from \citet{allard2012multi}, $\boldsymbol{y}_{i}$, for $i=1,2,\dots,n$, are assumed to have support on $(\mathcal{M},\mathcal{F},\mu)$, where $\mathcal{M}\subset \Re^{D}$, $\mathcal{F}$ is a $\sigma$-field defined on $\mathcal{M}$ and $\mu$ is a probability measure defined on $\mathcal{F}$. With $s=0,\dots,\infty$ denoting the scale index and $h=1,\dots,2^{s}$ denoting the node index within scale $s$, the binary clustering tree is defined as follows.
\begin{defn}
\noindent  A \textbf{binary clustering tree} of a 
metric measure space $(\mathcal{M},\mathcal{F},\mu)$ is a family
of open sets in $\mathcal{M}$, $\{Cell_{s,h}\}$, called dyadic cells,
such that

1. for every $s$, $\mu(\mathcal{M}\backslash\bigcup_{h=1}^{2^{s}}Cell_{s,h})=0$;

2. for $s\leq s'$ and $1\leq h'\leq2^{s'}$, either $Cell_{s',h'}\subseteq Cell_{s,h}$
or $\mu(Cell_{s',h'}\cap Cell_{s,h})=0$;

3. for $s<s'$ and $1\leq h'\leq2^{s'}$, there exists a unique $h=1,2,\dots,2^{s}$
such that $Cell_{s',h'}\subseteq Cell_{s,h}$.
\end{defn}

To learn the multiscale dictionaries, we implement the following three steps:
\begin{enumerate}
\item Obtain a binary clustering tree, $Cell_{s,h}$ for $s=0,\dots,\infty$ and $h=1,\dots,2^{s}$ using METIS \citep{karypis1998fast}, with the proximity matrix  computed using the approximate nearest neighbour (ANN) algorithm \citep{arya1998optimal}.
\item Estimate a $d$-dimensional affine approximation in each dyadic cell $Cell_{s,h}$ using fast rank-$d$ SVD \citep{rokhlin2009randomized}, yielding a local dictionary associated to this cell, denoted $\boldsymbol{\Phi}_{s,h}$.
\item Set $\boldsymbol{\mu}_{s,h}$ equal to the sample mean of $Cell_{s,h}$. 
\end{enumerate}

To illustrate these three steps, a 4\textendash{}level binary clustering tree of a synthetic parabola point cloud obtained using GMRA can be found in the appendix. The likelihood function for the general node $(s,h)$ is
\begin{equation}
\label{eq:SingleNode}
f_{s,h}(\boldsymbol{y}_i)=\mathcal{N}_D\big(\boldsymbol{y}_{i};\boldsymbol{\mu}_{s,h},\boldsymbol{\Phi}_{s,h}\boldsymbol{\Sigma}_{s,h}\boldsymbol{\Phi}_{s,h}^{T}+\sigma^{2}_{s}\boldsymbol{I}\big).
\end{equation}
With basic linear algebra we can write (\ref{eq:SingleNode}) as
\begin{align}
\begin{split}
f_{s,h}(\boldsymbol{y}_i) \propto
 & (\sigma_{s}^{2})^{-D/2}\prod_{m=1}^{d}u_{s,h,m}^{1/2}\exp\bigg\{-\frac{1}{2}\sigma_{s}^{-2} \times \\
 &\big[A_{s,h,i}-\sum_{m=1}^{d}(1-u_{s,h,m})(Z_{s,h,i}^{(m)})^{2}\big]\bigg\},
 \end{split}
\end{align}

where $A_{s,h,i}=\tilde{\boldsymbol{y}}_{i}^{T}\tilde{\boldsymbol{y}}_{i}$, $\tilde{\boldsymbol{y}}_{i}=\boldsymbol{y}_i-\boldsymbol{\mu}_{s,h}$, $u_{s,h,m}=(1+\sigma_{s}^{-2}\alpha_{s,h,m}^{2})^{-1}\mbox{, for }m=1,\dots,d$, and $\boldsymbol{Z}_{s,h,i}=\boldsymbol{\Phi}_{s,h}^{T}\tilde{\boldsymbol{y}}_{i}$, with $Z_{s,h,i}^{(m)}$ denoting its $m$th element. Details are reported in the appendix.

We first specify a prior for the ``full'' model where $d=D$. When $p$ is small, which we expect provides a good approximation in many applications, the information contained in the last $D-p$ columns of $\boldsymbol{\Phi}_{s,h}$ (columns of $\boldsymbol{\Phi}_{s,h}$ are ordered to be descending in their singular values) is negligible and treated as noise. We use a specially tailored prior that shrinks $\alpha_m^2$ to zero more aggressively as $m$ grows; this reduces MSE by pulling the small signals towards zero.  This is equivalent to shrinking $u_m$ increasingly for larger $m$. To accomplish this adaptive shrinkage, we propose a multiplicative exponential process prior that adapts the prior of \citet{bhattacharya2011sparse}, while placing an inverse-gamma prior on $\sigma_{s}^2$, for $s=0,\dots,\infty$:
\begin{align}
\label{eq:priors}
\begin{split}
\sigma_{s}^{-2} & \sim  \mbox{Ga}(a_{\sigma},b_{\sigma})\\ 
u_{s,h,m}       & \sim \mbox{Ga}_{(0,1)}(\delta_{s,h,m}+1,1)\\
\delta_{s,h,m}  & = \prod_{k=1}^{m}\tau_{s,h,k}\\
\tau_{s,h,k}    & \sim \mbox{Exp}_{[1,\infty)}(a)
\end{split}
\end{align}
where $\tau_{s,h,k}$, for $k=1,\dots,d$, are independent truncated exponential random variables, $\delta_{s,h,m}$ and $\tau_{s,h,m}$ are the global and the local shrinkage parameter for the $m$th column vector of $\boldsymbol{\Phi}_{s,h}$, respectively. Since $\tau_{s,h,k}\geq 1$ for $k=1,\dots,D$, $\delta_{s,h,m}=\prod_{k=1}^{m}\tau_{s,h,k}$ is increasing with respect to $m$. As a result, $u_{s,h,m}$ is stochastically approaching one since the truncated gamma density concentrates around one as $\delta_{s,h,m}$ increases.

However, for large $D$ it is wasteful to conduct computation for the full model, because as $m$ increases $u_{s,h,m}$ is shrunk very strongly to one, and the excess dimensions are effectively discarded.  Hence, we propose to truncate the model by setting $u_{s,h,m}=1$ ($\alpha^2_{s,h,m}=0$) for $m>d$, with $d$ an upper bound on the number of factors.  The following theorem shows that the  approximation error of the truncated prior decreases exponentially in $d$. The proof is reported in the appendix.
\begin{thm}
\label{ParaExpanupperboundthm} Assume $\boldsymbol{\Omega}_{s,h}=\boldsymbol{\Psi}\boldsymbol{\Sigma}_{s,h}\boldsymbol{\Psi}^{T}+\sigma_{s}^{2}\boldsymbol{I}$
 where $\boldsymbol{\Psi}$ is a orthonormal $D \times D$ matrix and $\boldsymbol{\Sigma}_{s,h}$ is a $D \times D$ positive diagonal matrix. The distributions of $\boldsymbol{\Sigma}_{s,h}$ and $\sigma_{s}^{2}$ are defined in (\ref{eq:priors}). Let $\boldsymbol{\Psi}^d$ denote the first $d$ columns of $\boldsymbol{\Psi}$, $\boldsymbol{\Sigma}_{s,h}^{d}=\diag{\alpha_{s,h,1}^2,\dots,\alpha_{s,h,d}^2}$ and let $\boldsymbol{\Omega}^{d}_{s,h}=\boldsymbol{\Psi}^d\boldsymbol{\Sigma}^{d}_{s,h}\big(\boldsymbol{\Psi}^d\big)^{T}+\sigma_{s}^{2}\boldsymbol{I}$. Then
for any $\epsilon>0$,
\[
Pr\{d_{\infty}(\boldsymbol{\Omega}_{s,h},\boldsymbol{\Omega}^{d}_{s,h})>\epsilon\}<\frac{6ba^{d}}{\epsilon(1-a)}
\]
for $d>2\log\{b/\epsilon(1-a)\}/\log(1/a)$, where $d_{\infty}(\boldsymbol{\Omega}_{s,h},\boldsymbol{\Omega}^{d}_{s,h})$ is defined as $\|\boldsymbol{\Omega}_{s,h}-\boldsymbol{\Omega}^{d}_{s,h}\|_{\infty}$. $\|A\|_{\infty}$ calculates the maximum absolute row sum of the matrix $A$, $b=E(\sigma_{s}^{2})$ and a = $E(\frac{1}{\tau_{s,h,1}})$.
\end{thm}

We then finish the formulation of GEODE by choosing a prior for the multiscale mixing weights $\pi_{s,h}$. This prior should be structured to allow adaptive learning of the appropriate tradeoff between coarse and fine scales. Heavily favoring coarse scales may lead to reduced variance but also high bias if the coarse scale approximation is not accurate.  High weights on fine scales may lead to low bias but high variance due to limited sample size in each fine resolution component.  With this motivation, \citet{CanaleDunson2014} proposed a  multiresolution stick-breaking process generalizing usual ``flat'' stick-breaking \citep{sethuraman1991constructive}.  In particular, let
\begin{equation}
\label{eq:SandR}
S_{s,h}\sim Be(1,a_{S})\mbox{, }R_{s,h}\sim Be(b_{R},b_{R})
\end{equation}
with $S_{s,h}$ denoting the probability that the observation stops at node $(s,h)$ of a binary tree and $R_{s,h}$ denoting the probability that the observation moves down to the right from node $(s,h)$ conditioning on not stopping at node $(s,h)$. Hence
\begin{equation}
\label{eq:pi}
\pi_{s,h}=S_{s,h}\prod_{r<s}(1-S_{r,g_{s,h,r}})T_{s,h,r}
\end{equation}
where $g_{s,h,r}=\left\lceil h/2^{s-r}\right\rceil $ denotes the ancestors of node $(s,h)$ at scale $r$, $T_{s,h,r}=R_{r,g_{s,h,r}}$ if node $(r+1,g_{s,h,r+1})$ is the right daughter of node$(r+1,g_{s,h,r})$ , otherwise $T_{s,h,r}=1-R_{r,g_{s,h,r}}$. \citet{CanaleDunson2014} showed that $\sum_{s=0}^{\infty}\sum_{h=1}^{2^{s}}\pi_{s,h}=1$ almost surely for any $a_{S},b_{R}>0$. This result makes the defined weights a proper set of multiscale mixing weights. As $a_{S}$ increases, finer scales are favored, resulting in a highly non-Gaussian density.

In practice, it is appealing to approximate the model by a finite-depth multiscale mixture. Let $L$ denote this depth and let $\{\tilde{\pi}_{s,h}\}_{s\leq L}$ denote the truncated weights, which are identical to $\{\pi_{s,h}\}$ except that the stopping probabilities at scale $L$ are set to be equal to one to ensure $\sum_{s=1}^{L}\sum_{h=1}^{2^{s}}\tilde{\pi}_{s,h}=1$. The accuracy of the approximation is discussed in the following theorem. The proof is reported in the appendix.
\begin{thm}
\label{ScaleTruncationBound}Let
\[
f^{L}(\boldsymbol{y}_{i})=\sum_{s=1}^{L}\sum_{h=1}^{2^{s}}\tilde{\pi}_{s,h}\mathcal{N}_{D}(\boldsymbol{y}_{i};\boldsymbol{\mu}_{s,h},\boldsymbol{\Phi}_{s,h}\boldsymbol{\Sigma}_{s,h}\boldsymbol{\Phi}_{s,h}^{T}+\sigma_{s}^{2}\boldsymbol{I})
\]
denote the approximation at scale $L$, let $P(B)=\int_{B}f(\boldsymbol{y}_{i})dy$
and $P^{L}(B)=\int_{B}f^{L}(\boldsymbol{y}_{i})dy$, for all $B \subset \Re^D$ denote
the probability measures corresponding to density $f(\boldsymbol{y}_{i})$ and $f^{L}(\boldsymbol{y}_{i})$.
Then we have,
\[
d_{TV}(P_{L},P)<\bigg(\frac{a_{S}}{1+a_{S}}\bigg)^{L},
\]
where $d_{TV}(P_{L},P)$ denotes the total variation distance between $P_{L}(B)$ and $P(B)$.
\end{thm}
The above theorem indicates that the approximation error decays at an exponential rate.

\subsection{Posterior Computation}
\label{ssec:PosteriorCompu}

The usual frequentist method of selecting an upperbound $d$ thresholds the singular values, leading to substantial sensitivity to threshold choice.  For large $D$, the upper bound $d$ has to be chosen in advance so that fast rank-d SVD can be achieved \citep{rokhlin2009randomized}. Typically, conservative choice for $d$ is implemented in order to ensure $d\geq p$, adding a burden to both computation and storage. We avoid this by automatically deleting redundant dictionary elements, and hence decreasing $d$, as computation proceeds. To this end we adopt an adaptive Gibbs sampler similar to that developed  by \citet{bhattacharya2011sparse}. The adaptive Gibbs sampler randomly deletes redundant dimensions at $t$th iteration according to probability $p(t)=\exp(c_{0}+c_{1}t)$. The values of $c_{0}$ and $c_{1}$ are chosen to ensure frequent adaption at the beginning of the chain and an exponentially fast decay in frequency after that. We fix $c_{0}=-1$, $c_{1}=-0.005$ and $tol=10^{-4}$ as default, where $tol$ is a prespecified threshold.

Introduce the membership variables $\big(s_{i},h_{i}\big)$, then the conditional posterior is given by
\begin{align}
\label{eq:UpdateMembership}
\begin{split}
p(s_{i} = s,h_{i}=h) \propto
& \pi_{s,h} \mathcal{N}_{D}(\pmb{\mu}_{s,h},\\
& \boldsymbol{\Phi}_{s,h}\boldsymbol{\Sigma}_{s,h}\boldsymbol{\Phi}_{s,h}^{T}+\sigma_{s}^{2}\boldsymbol{I}).
\end{split}
\end{align}
A multiscale slice sampler \citep{CanaleDunson2014} could save computation when $L$ is large. The conditional posteriors of $S_{s,h}$ and $R_{s,h}$ are given by
\begin{align}
\label{eq:UpdateSH}
\begin{split}
S_{s,h} & \sim \mbox{Beta}(1+n_{s,h},a_{S}+v_{s,h}-n_{s,h}),\\
R_{s,h} & \sim \mbox{Beta}(b_{R}+r_{s,h},b_{R}+v_{s,h}-n_{s,h}-r_{s,h}),
\end{split}
\end{align}
where $v_{s,h}$ is the number of observations passing through node $(s,h)$, $n_{s,h}$ is the number of observations stopping at node $(s,h)$, and $r_{s,h}$ is the number of observations that continue to the right after passing through node $(s,h)$. The slice sampler contributes to the computation by allowing the allocation to take place in a subset of all scales of the tree, which can be efficient when we have a deep tree structure. Let $\mathcal{D}_{s,h}$ denote the set of deleted dimension indices (the deleted pool) of node $(s,h)$ and $\mathcal{R}_{s,h}$ denote the set of retained dimension indices (the remaining pool) of node $(s,h)$. Combining all the techniques discussed above, the Bayesian GEODE algorithm can be summarized as follows

\textbf{The first stage}:
\begin{enumerate}
\item Compute a multiscale dictionary $\{\pmb{\Phi}_{s,h},\pmb{\mu}_{s,h}\}$ using GMRA and initialize the algorithm.
\end{enumerate}
\textbf{The second stage}, iterate until the desired posterior sample size:
\begin{enumerate}
\item Update $s_{i}$ and $h_{i}$ for all $i$ according to (\ref{eq:UpdateMembership}).
\item Update $S_{s,h}$ and $R_{s,h}$ for all $s$ and $h$ according to  (\ref{eq:UpdateSH}).
\item Update $u_{s,h,m}$ for all $s$, $h$ and $m$ according to $\mbox{Gamma}_{(0,1)}\big(\hat{a}_{s,h,m},\hat{b}_{s,h,m}\big)$, where $\hat{a}_{s,h,m} = \prod_{k=1}^{m}\tau_{s,h,k}+n_{s,h}/2$ and $\hat{b}_{s,h,m} = 1+\frac{1}{2}\sigma_{s}^{-2}\sum_{y_{i}\in C_{s,h}}(\boldsymbol{Z}_{s,h,i}^{(m)})^{2}$.
\item Update $\tau_{s,h,m}$ for all $s$, $h$ and $m$ according to $\mbox{Exp}_{[1,\infty)}\big(\hat{\lambda}_{s,h,m}\big)$, where $\hat{\lambda}_{s,h,m} = a_{\tau}-\ln(\prod_{j>m-1}u_{s,h,j})$
\item Update $\sigma_{s}^{-2}$ for all $s$ according to
$\mbox{Gamma}\big(\hat{c}_{s}, \hat{d}_{s}\big)$, where $\hat{c}_{s} = a_{\sigma}+Dn_{s}/2$, $\hat{d}_{s} = \frac{1}{2}\sum_{y_{i}\in C_{s}}\big[A_{s,h,i}-\sum_{j=1}^{d}(1-u_{s,h,j})(\boldsymbol{Z}_{s,h,i}^{(j)})^{2}\big]+b_{\sigma}$, $C_{s}$ denotes the set of observations stopping at scale $s$, and $n_{s}$ denotes the size of $C_{s}$.
\item Compute $p(t)=\exp(c_{0}+c_{1}t)$, generate $g$ from Uniform$(0,1)$. If $g > p(t)$, go back to step 2 until the desired iteration number.
\item For all $(s,h)$, compute $r_{s,h,m}^{t}=\big(\alpha_{s,h,m}^{t}\big)^{2}/\max_{j\in\mathcal{R}_{s,h}}{\big(\alpha_{shj}^{t}\big)^{2}}$, for $m \in \mathcal{R}_{s,h}$. Remove all $m$ from $\mathcal{R}_{s,h}$ to $\mathcal{D}_{s,h}$ if $r_{s,h,m}^{t}<tol$. If no such $m$ exists, then randomly add back one dimension $m$ from $\mathcal{D}_{s,h}$ to $\mathcal{R}_{s,h}$ according to $p(m)\propto I_{m\in\mathcal{D}_{s,h}}r_{s,h,m}^{t-1}$.
\end{enumerate}

The derivation of all the conditional posteriors can be found in the supplement. Through the paper, we fix $a_{\sigma}=1/2$, $b_{\sigma}=1/2$ and $a=0.05$, and use the default parameters in the GMRA code provided by \citet{allard2012multi}.

\subsection{Missing Data Imputation} \label{SSec:MDI}
Bayesian models better utilize the partially observed data by probabilistically imputing the missing features based on its conditional posterior distribution. Notations $\boldsymbol{y}_{M}$ and $\boldsymbol{y}_{O}$ are introduced as the missing part and the observed part of $\pmb{y}$ respectively. Similarly, slightly abusing the notations, 
let $\boldsymbol{\mu}_{M}$ and $\boldsymbol{\Phi}_{M}$ denote the missing parts of $\pmb{\mu}_{s,h}$ and $\pmb{\Phi}_{s,h}$, and let  $\boldsymbol{\mu}_{O}$ and $\boldsymbol{\Phi}_{O}$ denote the observed parts. The following proposition enables efficient sampling from the conditional posterior distribution  $p(\boldsymbol{y}_{M}|\boldsymbol{y}_{O},\boldsymbol{\Theta})$, where $\boldsymbol{\Theta}$ denotes the all unknown parameters in the model. The computational analysis is provided in \S~\ref{Sec: CA} and simulation studies are provided in \S~\ref{Sec:SimulationStudies}.
\begin{prop}
\label{Prop MD}
For node $(s,h)$, introduce augmented data $\boldsymbol{\eta}_{i}$ such that $(\boldsymbol{y}_{i}|\boldsymbol{\eta}_{i},\boldsymbol{\Theta},s_{i}=s,h_{i}=h) \sim \mathcal{N}_{D}(\boldsymbol{\mu}_{s,h}+\boldsymbol{\Phi}_{s,h}\boldsymbol{\eta}_{i},\sigma_{s}^{2}\boldsymbol{\mbox{I}}_{D})$ and $(\boldsymbol{\eta}_{i}|\boldsymbol{\Theta},s_{i}=s,h_{i}=h) \sim \mathcal{N}_{d}(0,\boldsymbol{\Sigma}_{s,h})$, for $i=1,\dots,n$. Then
we have the conditional distribution with $\boldsymbol{\eta}_{i}$ marginalized out equal $(\boldsymbol{y}_{i}|\boldsymbol{\Theta},s_{i}=s,h_{i}=h) \sim \mathcal{N}_{D}(\boldsymbol{\mu}_{s,h}+\boldsymbol{\Phi}_{s,h}\boldsymbol{\Sigma}_{s,h}\boldsymbol{\Phi}_{s,h}^{T},\sigma_{s}^{2}\boldsymbol{\mbox{I}}_{D})$. Furthermore, conditional on $s_{i}=s$ and $h_{i}=h$ we have
\begin{align*}
\label{eq:prop}
\boldsymbol{\eta}_{i}|\boldsymbol{y}_{O},\boldsymbol{\Theta} & \sim \mathcal{N}_{d}(\hat{\boldsymbol{\mu}}_{\eta},\hat{\boldsymbol{\Sigma}}_{\eta}), \\
\boldsymbol{y}_{M}|\boldsymbol{\eta}_{i},\boldsymbol{y}_{O},\boldsymbol{\Theta} & \sim \mathcal{N}_{m_{i}}(\boldsymbol{\mu}_{M}+\boldsymbol{\Phi}_{M}\boldsymbol{\eta}_{i},\sigma_{s}^{2}\mbox{I}_{m_{i}}),
\end{align*}
where $\hat{\boldsymbol{\Sigma}}_{\eta}=\big(\boldsymbol{\Sigma}_{s,h}\boldsymbol{\Phi}_{O}^{T}\boldsymbol{\Phi}_{O}/\sigma_{s}^{2}+\mbox{I}\big)^{-1}\boldsymbol{\Sigma}_{s,h}$
and $\hat{\boldsymbol{\mu}}_{\eta}=\hat{\boldsymbol{\Sigma}}_{\eta}\boldsymbol{\Phi}_{O}^{T}(\boldsymbol{y}_{O}-\boldsymbol{\mu}_{O})/\sigma_{s}^{2}$.
\end{prop}
The proposition also provides an efficient way to predict multivariate response, which is applied to image inpainting in \S~\ref{SSec:ImageInpainting}. Proof is reported in the appendix.

\section{Computational Aspects} \label{Sec: CA}
\begin{figure}
\begin{centering}
\begin{minipage}[t]{0.45\textwidth}%
\begin{center}
\includegraphics[width=1\textwidth]{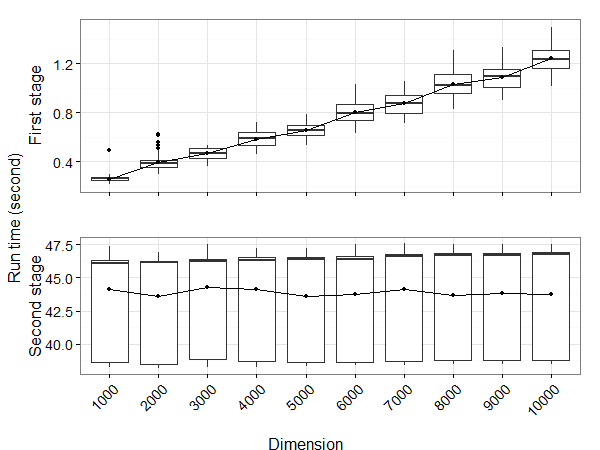}
\par\end{center}%
\end{minipage}
\par\end{centering}
\caption{Boxplot of the computational times of 100 replicate experiments at different ambient dimensions, with means jointed by segments.}
\label{fig:CTime}
\end{figure}
When data are complete, the computational cost of our implementation of GMRA is $O\big(nD(\log{n}+d^2)\big)$ \citep{allard2012multi}. The cost of computing the sufficient statistics $\big\{A_{s,h,i}\big\}$, $\big\{Z_{s,h,i}\big\}$ is $O\big(nD2^{L}d\big)$. Hence the overall cost of the first stage is given by
\[O\bigg(nD(\log{n}+d^2+2^{L}d)\bigg),\]
which only increases linearly in $D$. Letting $T$ be the total iteration number of the Gibbs sampler, the overall computational cost of the second stage  is given by
\begin{equation*}
O\bigg(T(2^{L}d^{3}+nd)\bigg),
\end{equation*}
which is independent of $D$.

When data has missing features, with $\boldsymbol{\Phi}_{O}^{T}\boldsymbol{\Phi}_{O}$ and $\boldsymbol{\Phi}_{O}^{T}(\boldsymbol{y}_{i}^{O}-\boldsymbol{\mu}_{O})$ stored as sufficient statistics, the computational cost of the first stage is given by
\[
O\bigg(nD(\log{n}+d^2+2^{L}d)+n_{m}Dd^2\bigg),
\]
and the cost of the second stage is given by
\[
O\bigg(T(2^{L}d^{3}+nd+n_{m}M)\bigg),
\]
where $n_{m}$ denotes the number of partial observations and $M=\max_{i=1,\dots,n}{m_{i}}$.

The computation time of the complete case is reported in Figure \ref{fig:CTime}, where 100 random samples were generated by projecting a 3-D Swissroll into higher dimensional ambient spaces. $d$ was set to be 10. The linearity in the first stage and the independence in the second stage with respect to D can be easily seen.

Differently from GEODE, traditional Bayesian MFA models have to learn and store the $D \times d$ factor loading matrices within each iteration in the MCMC, making both the computation and the storage daunting tasks when $D$ is very large. Moreover, due to the reduced number of parameters and lower posterior dependence in these parameters, our Gibbs sampler for GEODE converges and mixes dramatically faster than MCMC algorithms for fully Bayesian MFA models. This reduces the number of samples needed; we run the sampler 1,000 iterations, with the first 500 as a burn-in.  Experimental results show convergence typically occurs very fast. 

Note that all data experiments in the paper were run in matlab version 2012a on a x86\_64 linux machine with a $8\times 3.40$ GHz Intel(R) Core(TM) i7-3770 processor. Furthermore, note that our Gibbs sampler is written in matlab and hence the computing time of the second stage could be greatly reduced using lower level languages.

\begin{figure}
\begin{centering}
\begin{minipage}[t]{0.45\textwidth}%
\begin{center}
\includegraphics[width=1\textwidth]{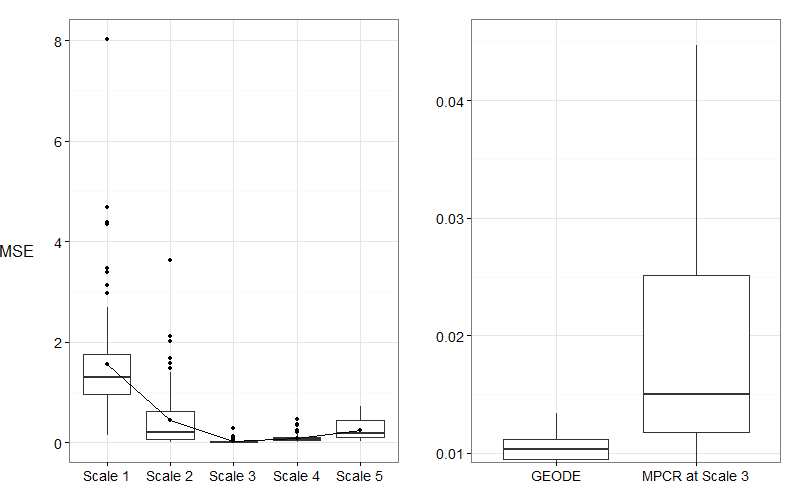}
\par\end{center}%
\end{minipage}
\par\end{centering}
\caption{\textbf{Left}: Boxplot of predictive MSE of MPCR at different scales; \textbf{Right}: Predictive MSE of GEODE compared with the best MPCR can do.}
\label{fig:Smooth}
\end{figure}

\section{Simulation Studies} \label{Sec:SimulationStudies}
To demonstrate GEODE, several simulation studies were conducted. Our aim is to highlight several characteristics of the approach: the improved quality by mixing over different scales, the ability to learn the true intrinsic dimension or a tight upperbound, the ability to impute missing data and the accurate characterization of uncertainty. Through the simulation studies, $d$ was set to be 10, providing an upper bound on the intrinsic dimension.  The method is not sensitive to the choice of this upper bound.
\subsection{Smoothness Adaptation} \label{SSec:Mix}
By mixing across different scales, GEODE is able to tradeoff between coarser scales and finer scales in a Bayesian manner adapting to the local smoothness.  To see this, a multi-scale principal component regression (MPCR) based on GMRA is proposed and compared with GEODE. The MPCR, being a natural combination of GMRA and principal component regression (PCR), learns local regression coefficients by applying PCR to subsets of observations at each node within a specific level. The prediction is made by first assigning the data point to the node  closest to this data point in terms of Euclidean distance to the center, and then predicting using the local regression coefficients. It is a natural comparison to GEODE since both use the same binary tree structure and the same multiscale dictionaries. MPCR predicts based on all the nodes within a specific scale while GEODE mixes over all scales.
\begin{figure}[H]
\begin{centering}
\begin{minipage}[t]{0.45\textwidth}%
\begin{center}
\includegraphics[width=1\textwidth]{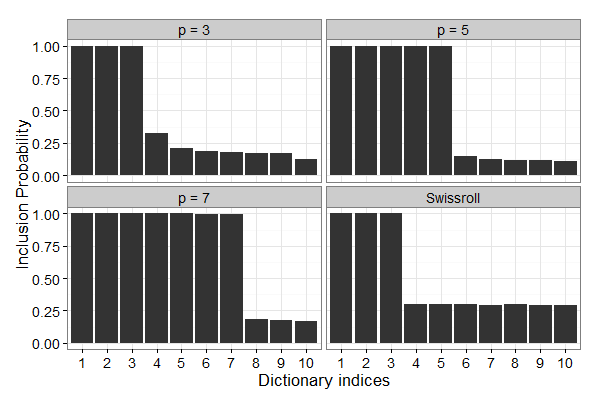}
\par\end{center}%
\end{minipage}
\par\end{centering}
\caption{Average inclusion probabilities for each dimension under different scenarios, with 10 being an upper bound.}
\label{fig:InlProb}
\end{figure}
In the simulation study, 100 independent samples with 1100 observations were generated from a mixture of three Gaussians with $D=10000$, whose intrinsic dimensions equal $3$, $5$ and $7$ respectively. In each sample, 1000 observations were randomly selected to train the model, and the other 100 were used as test data. One dimension of the test data is assumed to be missing and to be predicted. Performance of GEODE is compared with that of MPCR in terms of the mean square prediction error, which is shown in Figure \ref{fig:Smooth}. The MSE curve of MPCR is u-shaped, indicating over-fitting at fine scale. MDLR clearly outperforms MPCR even under ideal conditions for MPCR. This suggests that GEODE efficiently utilized the local smoothness information , while adaptively borrowing information across scales.

\subsection{Intrinsic Dimension Learning} \label{SSec:IDL}
The adaptive Gibbs sampler automatically excludes unnecessary dimensions.  The posterior mean inclusion probabilities are useful in estimating the true intrinsic dimension.  These probabilities were computed under each simulation case, with results shown in Figure \ref{fig:InlProb}. In the Gaussian mixture case, GEODE successfully learned the true p, with the redundant dimensions excluded with more than 70\% probability, saving computation and storage.  For the Swissroll example, GEODE instead provided a tight upperbound for the true $p$.

\subsection{Regression With Missing Data} \label{SSec:Regression}

In this simulation study, 100 independent samples are generated from 9 different scenarios involving Gaussian or manifold (Swissroll) data, different ambient dimensions $D$ and different intrinsic dimensions $p$. GEODE was compared with competing methods in regression problems either with or without missing data. Scenarios 1 - 6 are linear Gaussian data and scenarios 7-9 are Swissroll data embedded in high dimensional ambient spaces. Simulation details are reported in the appendix. 

For Gaussian data, our method is compared with elastic net (EN) and PCR. For Swissroll data, our method is compared with random forest (RF). To make the computation of PCR and RF feasible for our studies, fast rank-k SVD was applied in both cases. RF was applied after the data have been projected to a 10 dimensional space using fast SVD. As can be seen from Figure \ref{fig:MSE}, GEODE has a consistently better predictive accuracy than the competing methods. Moreover, GEODE successfully imputed the missing data while maintaining similar MSE in the presence of missing data, while methods that discard observations with missing data have clearly increased MSE. Empirical 95\% coverages of intervals out of sample are presented in Figure \ref{fig:CovInclu}. As can be seen, GEODE only slightly underestimated uncertainty.

The results demonstrated the capability of the proposed method to properly characterize uncertainty and impute missing data, while maintaining computational efficiency and accurate predictions.
\begin{figure}
\begin{centering}
\begin{minipage}[t]{0.45\textwidth}%
\begin{center}
\includegraphics[width=1\textwidth]{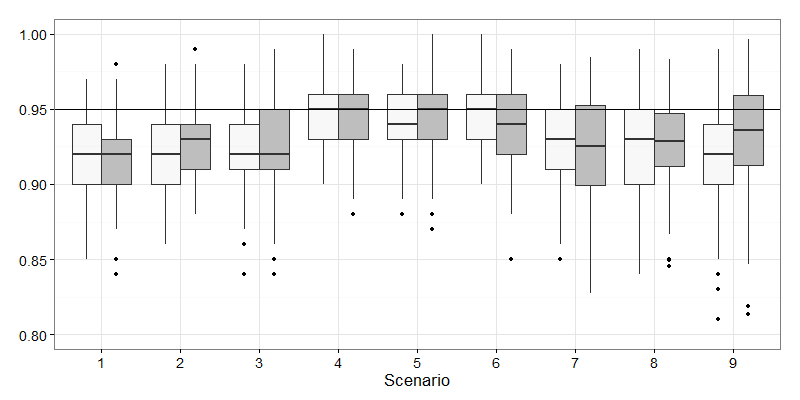}
\par\end{center}%
\end{minipage}
\par\end{centering}
\caption{Boxplots of the empirical coverages of 100 replicate experiments, with fully observed datasets denoted by light grey and partially observed ones denoted by dark grey.}
\label{fig:CovInclu}
\end{figure}

\section{Application} \label{Sec:App}
GEODE is further demonstrated first in a multivariate response regression application and then in a supervised classification problem. In both applications, $d=20$. Increasing $d$ moderately had essentially no impact on the results.
\subsection{Image Inpainting}
\label{SSec:ImageInpainting}
The Frey faces data \citep{roweis2002global} contains 1965 $20\times28$ video frames of a single face with different expressions. Conducting the same experiment as done by \citet{titsias2010bayesian}, the data set is randomly split into 1000 training images and 965 testing images with a random half of the pixels missing. GEODE was trained for less than 2 minutes, and reconstruction  (prediction) of all 965 testing images was done in less than 10 minutes. The mean absolute reconstruction error of GEODE is 7.04, which outperforms the error of 7.40 reported by \citet{titsias2010bayesian}. 10 randomly selected reconstructions are shown on the left in Figure \ref{fig:Inpainting}, with 4 manually designed missingness cases shown on the right. GEODE also outperforms the results shown by \citet{adams-wallach-ghahramani-2010a} by looking at their visualized results.  It is also noted that \citet{adams-wallach-ghahramani-2010a} reported a few hours of computational time in reconstructing 100 images based on 1865 training images.

\begin{figure*}[t]
\begin{centering}
\begin{minipage}[t]{0.8\textwidth}%
\begin{center}
\includegraphics[width=1\textwidth]{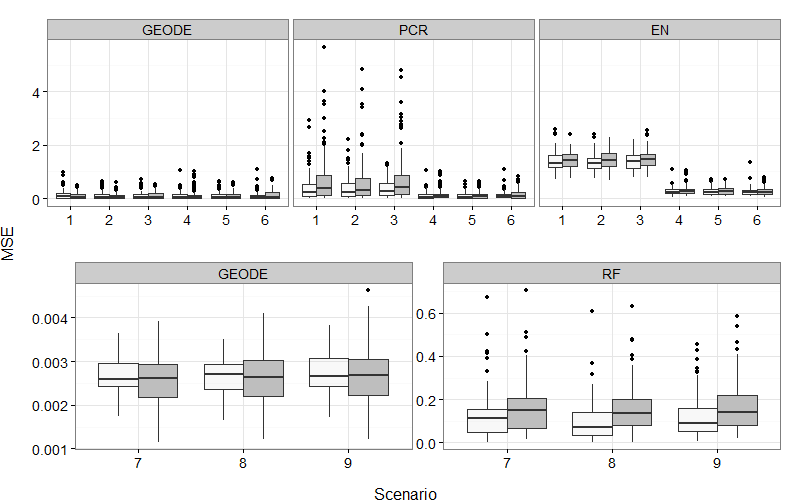}
\par\end{center}%
\end{minipage}
\par\end{centering}
\caption{Comparison of performance between GEODE and other methods with respect to MSE, the vertical bars represent the 95\% empirical intervals, with fully observed datasets denoted by light grey and partially observed ones denoted by dark grey.}
\label{fig:MSE}
\end{figure*}

\begin{figure*}
\begin{centering}
\begin{minipage}[t]{0.8\textwidth}%
\begin{center}
\includegraphics[width=1\textwidth]{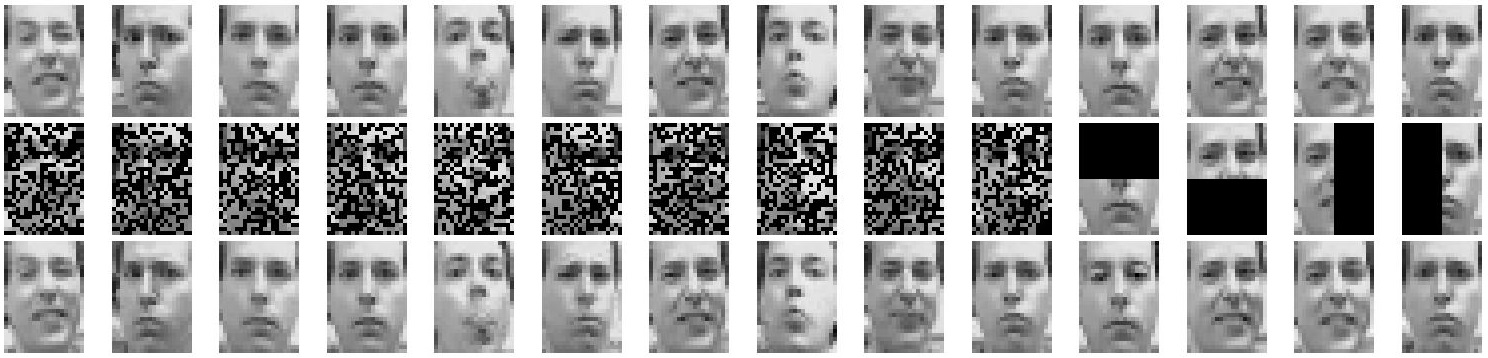}
\par\end{center}%
\end{minipage}
\par\end{centering}
\caption{The first row shows the original images, second row shows the images with pixels missing, and the third row shows the reconstructed images.}
\label{fig:Inpainting}
\end{figure*} 
\subsection{Digit Classification}\label{SSec:Digit}
GEODE was used as a probabilistic classifier for the MNIST handwritten data, which contains 70000 $28\times 28$ grey scale handwritten digits images. First, one GEODE was trained for each of the 10 digits over a total of 60000 training data for around 90 minutes. Then within each iteration of the Gibbs sampler, the 10 GEODE's worked in a Naive Bayes way and generated a likely class. The ``voting'' process took 7 minutes for 10000 testing images and the mode of these votes were computed as the classification results. The classification error was 2.32\%.

\section{Discussion}
\label{Sec:Discussion}

In many applications, high-dimensional data with unknown joint distribution are collected. Despite the dramatic importance of learning the joint distribution of such data, few probabilistic methods that scale well to high-dimension and provide an adequate characterization of uncertainty are available.   Bayesian nonparametric methods based on mixtures of multivariate Gaussian kernels are widely used, but face major bottlenecks in scaling to higher dimensions.  To tackle this problem, we proposed an empirical Bayes density estimator combining manifold learning and Bayesian nonparametric density estimation.  One of the building blocks of our method focuses on single Gaussian factor decomposition in which variables are linearly related, showing excellent performance in scaling computationally and in generalization error, while providing a valid characterization of uncertainty in predictions. The other building block is a multiscale mixture generalization, which accommodates unknown density, nonlinear relationships and nonlinear subspaces. This approach showed excellent performance in inferring the subspace dimension, estimating the subspace, and characterizing the joint density of the data in the ambient space.  The proposed methods are broadly applicable to many learning problems including regression or classification with missing features.

\bibliographystyle{plainnat}
\bibliography{references}

\begin{thebibliography}{23}
\providecommand{\natexlab}[1]{#1}
\providecommand{\url}[1]{\texttt{#1}}
\expandafter\ifx\csname urlstyle\endcsname\relax
  \providecommand{\doi}[1]{doi: #1}\else
  \providecommand{\doi}{doi: \begingroup \urlstyle{rm}\Url}\fi

\bibitem[Adams et~al.(2010)Adams, Wallach, and
  Ghahramani]{adams-wallach-ghahramani-2010a}
R.~P. Adams, H.~M. Wallach, and Z.~Ghahramani.
\newblock Learning the structure of deep sparse graphical models.
\newblock \emph{Journal of Machine Learning Research: Workshop and Conference
  Proceedings (AISTATS)}, 9:\penalty0 1--8, 05/2010 2010.

\bibitem[Allard et~al.(2012)Allard, Chen, and Maggioni]{allard2012multi}
W.~K. Allard, G.~Chen, and M.~Maggioni.
\newblock Multi-scale geometric methods for data sets ii: {G}eometric
  multi-resolution analysis.
\newblock \emph{Applied and Computational Harmonic Analysis}, 32\penalty0
  (3):\penalty0 435--462, 2012.

\bibitem[Arya et~al.(1998)Arya, Mount, Netanyahu, Silverman, and
  Wu]{arya1998optimal}
S.~Arya, D.~M. Mount, N.~S. Netanyahu, R.~Silverman, and A.~Y. Wu.
\newblock An optimal algorithm for approximate nearest neighbor searching fixed
  dimensions.
\newblock \emph{Journal of the ACM (JACM)}, 45\penalty0 (6):\penalty0 891--923,
  1998.

\bibitem[Bhattacharya and Dunson(2011)]{bhattacharya2011sparse}
A.~Bhattacharya and D.~B. Dunson.
\newblock Sparse {B}ayesian infinite factor models.
\newblock \emph{Biometrika}, 98\penalty0 (2):\penalty0 291--306, 2011.

\bibitem[Canale and Dunson(2014)]{CanaleDunson2014}
A.~Canale and D.~B. Dunson.
\newblock Multiscale {B}ernstein polynomials for densities.
\newblock 2014.
\newblock {\tt arXiv:1410.0827 [stat.ME]}.

\bibitem[Carvalho et~al.(2008)Carvalho, Chang, Lucas, Nevins, Wang, and
  West]{carvalho2008high}
C.~M. Carvalho, J.~Chang, J.~E. Lucas, J.~R. Nevins, Q.~Wang, and M.~West.
\newblock High-dimensional sparse factor modeling: applications in gene
  expression genomics.
\newblock \emph{Journal of the American Statistical Association}, 103\penalty0
  (484), 2008.

\bibitem[Diebolt and Robert(1994)]{diebolt1994estimation}
J.~Diebolt and C.~P. Robert.
\newblock Estimation of finite mixture distributions through bayesian sampling.
\newblock \emph{Journal of the Royal Statistical Society. Series B (Statistical
  Methodology)}, 56\penalty0 (2):\penalty0 363--375, 1994.

\bibitem[Escobar and West(1995)]{escobar1995bayesian}
M.~D. Escobar and M.~West.
\newblock Bayesian density estimation and inference using mixtures.
\newblock \emph{Journal of the American Statistical Association}, 90\penalty0
  (430):\penalty0 577--588, 1995.

\bibitem[Ghahramani and Beal(1999)]{ghahramani1999variational}
Z.~Ghahramani and M.~J. Beal.
\newblock Variational inference for {B}ayesian mixtures of factor analysers.
\newblock In \emph{NIPS}, volume~12, pages 449--455, 1999.

\bibitem[Ghahramani et~al.(1996)Ghahramani, Hinton,
  et~al.]{ghahramani1996algorithm}
Z.~Ghahramani, G.~E. Hinton, et~al.
\newblock The {EM} algorithm for mixtures of factor analyzers.
\newblock Technical report, CRG-TR-96-1, University of Toronto, 1996.

\bibitem[Karypis and Kumar(1998)]{karypis1998fast}
G.~Karypis and V.~Kumar.
\newblock A fast and high quality multilevel scheme for partitioning irregular
  graphs.
\newblock \emph{SIAM Journal on Scientific Computing}, 20\penalty0
  (1):\penalty0 359--392, 1998.

\bibitem[Lawrence(2005)]{lawrence2005probabilistic}
N.~Lawrence.
\newblock Probabilistic non-linear principal component analysis with {G}aussian
  process latent variable models.
\newblock \emph{The Journal of Machine Learning Research}, 6:\penalty0
  1783--1816, 2005.

\bibitem[Liu et~al.(2007)Liu, Lafferty, and Wasserman]{liu2007sparse}
H.~Liu, J.~D. Lafferty, and L.~A. Wasserman.
\newblock Sparse nonparametric density estimation in high dimensions using the
  rodeo.
\newblock In \emph{International Conference on Artificial Intelligence and
  Statistics}, pages 283--290, 2007.

\bibitem[Rasmussen(1999)]{rasmussen1999infinite}
C.~E. Rasmussen.
\newblock The infinite {G}aussian mixture model.
\newblock In \emph{NIPS}, volume~12, pages 554--560. MIT; 1998, 1999.

\bibitem[Richardson and Green(1997)]{richardson1997bayesian}
S.~Richardson and P.~J. Green.
\newblock On {B}ayesian analysis of mixtures with an unknown number of
  components (with discussion).
\newblock \emph{Journal of the Royal Statistical Society: Series B (Statistical
  Methodology)}, 59\penalty0 (4):\penalty0 731--792, 1997.

\bibitem[Rokhlin et~al.(2009)Rokhlin, Szlam, and Tygert]{rokhlin2009randomized}
V.~Rokhlin, A.~Szlam, and M.~Tygert.
\newblock A randomized algorithm for principal component analysis.
\newblock \emph{SIAM Journal on Matrix Analysis and Applications}, 31\penalty0
  (3):\penalty0 1100--1124, 2009.

\bibitem[Roweis et~al.(2002)Roweis, Saul, and Hinton]{roweis2002global}
S.~T. Roweis, L.~K. Saul, and G.~E. Hinton.
\newblock Global coordination of local linear models.
\newblock In \emph{NIPS}, volume~2, pages 889--896. MIT; 1998, 2002.

\bibitem[Sethuraman(1994)]{sethuraman1991constructive}
J.~Sethuraman.
\newblock A constructive definition of {D}irichlet priors.
\newblock \emph{Statistica Sinica}, 4:\penalty0 639--650, 1994.

\bibitem[Shen et~al.(2013)Shen, Tokdar, and Ghosal]{shen2011adaptive}
W.~Shen, S.~T. Tokdar, and S.~Ghosal.
\newblock Adaptive {B}ayesian multivariate density estimation with {D}irichlet
  mixtures.
\newblock \emph{Biometrika}, 100\penalty0 (4):\penalty0 623--640, 2013.

\bibitem[Tenenbaum et~al.(2000)Tenenbaum, Silva, and
  Langford]{tenenbaum2000global}
J.~B. Tenenbaum, V.~De Silva, and J.~C. Langford.
\newblock A global geometric framework for nonlinear dimensionality reduction.
\newblock \emph{Science}, 290\penalty0 (5500):\penalty0 2319--2323, 2000.

\bibitem[Titsias and Lawrence(2010)]{titsias2010bayesian}
M.~Titsias and N.~Lawrence.
\newblock Bayesian {G}aussian process latent variable model.
\newblock In \emph{the International Conference on Articial Intelligence and
  Statistics}, 2010.

\bibitem[Vincent and Bengio(2003)]{vincent2003manifold}
P.~Vincent and Y.~Bengio.
\newblock Manifold parzen windows.
\newblock In \emph{NIPS}, volume~15, pages 849--856. MIT; 1998, 2003.

\bibitem[Wang and Titterington(2004)]{wang2004inadequacy}
B.~Wang and M.~Titterington.
\newblock Inadequacy of interval estimates corresponding to variational
  {B}ayesian approximations.
\newblock In \emph{10th Int. Workshop Artific. Intell. Statist. Ed. R.G.
  Cowell}, 2004.

\end{thebibliography}

\onecolumn
\appendix
\begin{center}
\huge{\textbf{Appendix}}
\end{center}
\section{Formulation}

To illustrate the binary clustering tree, a 4\textendash{}level binary clustering tree of a synthetic parabola point cloud obtained using GMRA can be found in Figure \ref{fig:TreeStructureDescription}.

The likelihood function of GEODE can be written as
\begin{align}
\label{eq:Likelihood}
\tag{A1}
\begin{aligned}
f_{s,h}(\boldsymbol{y}_i) \propto 
&(\sigma_{s}^{2})^{-D/2}\prod_{m=1}^{d}u_{s,h,m}^{1/2}\times \exp\bigg\{-\frac{1}{2}\sigma_{s}^{-2}\\
&\big[A_{s,h,i}-\sum_{m=1}^{d}(1-u_{s,h,m})(Z_{s,h,i}^{(m)})^{2}\big]\bigg\},
\end{aligned}
\end{align}
which can be derived using the following two propositions.
\begin{prop}
\label{Quadratic}$\boldsymbol{\Sigma}=diag(\alpha_{1}^{2},\dots,\alpha_{d}^{2})$
is a $d\times d$ matrix with all diagonal
entries larger than $0$, $\boldsymbol{\Phi}$ is a $D\times d$ orthonormal matrix,
we have,
\begin{equation*}
(\sigma^{2}\boldsymbol{I}+\boldsymbol{\Phi}\boldsymbol{\Sigma}\boldsymbol{\Phi}{}^{T})=\sigma^{-2}\boldsymbol{I}-\sigma^{-4}\boldsymbol{\Phi}\tilde{\boldsymbol{\Sigma}}\boldsymbol{\Phi}^{T},
\end{equation*}
where $\tilde{\boldsymbol{\Sigma}}=diag(\frac{\alpha_{1}^{2}}{1+\sigma^{-2}\alpha_{1}^{2}},\frac{\alpha_{2}^{2}}{1+\sigma^{-2}\alpha_{2}^{2}},\dots,\frac{\alpha_{d}^{2}}{1+\sigma^{-2}\alpha_{d}^{2}})$.\end{prop}
\begin{proof}
By the orthonormality of the dictionary, we have $\boldsymbol{\Phi}^{T}\boldsymbol{\Phi}=\boldsymbol{I}_{d}$. And by the matrix inversion formula, 
\begin{eqnarray}
(\sigma^{2}I+\boldsymbol{\Phi}\boldsymbol{\Sigma}\boldsymbol{\Phi}^{T})^{-1} & = & \sigma^{-2}\boldsymbol{I}-\sigma^{-4}\boldsymbol{\Phi}(\boldsymbol{I}+\sigma^{-2}\boldsymbol{\Sigma}\boldsymbol{\Phi}^{T}\boldsymbol{\Phi})^{-1}\boldsymbol{\Sigma}\boldsymbol{\Phi}^{T}\nonumber \\
 & = & \sigma^{-2}\boldsymbol{I}-\sigma^{-4}\boldsymbol{\Phi}(\boldsymbol{I}+\sigma^{-2}\boldsymbol{\Sigma})^{-1}\boldsymbol{\Sigma}\boldsymbol{\Phi}^{T}\nonumber \\
 & = & \sigma^{-2}\boldsymbol{I}-\sigma^{-4}\boldsymbol{\Phi}\tilde{\boldsymbol{\Sigma}}\boldsymbol{\Phi}^{T}\nonumber
\end{eqnarray}
\end{proof}
\begin{prop}
\label{determ}Under the same setting of Proposition  \ref{Quadratic},
we have
\begin{equation*}
|\sigma^{2}\boldsymbol{I}+\boldsymbol{\Phi}\boldsymbol{\Sigma}\boldsymbol{\Phi}^{T}|^{-1/2}=(\sigma^{2})^{-D/2}\prod_{m=1}^{d}(\frac{1}{1+\sigma^{-2}\alpha_{m}^{2}})^{1/2}.
\end{equation*}
\end{prop}
\begin{proof}
By Theorem Schur's formula,
\begin{eqnarray}
|\sigma^{2}\boldsymbol{I}+\boldsymbol{\Phi}\boldsymbol{\Sigma}\boldsymbol{\Phi}^{T}|^{-1/2} & = & (\sigma^{2})^{-D/2}|\boldsymbol{I}_{D}+\sigma^{-2}\boldsymbol{\Phi}\boldsymbol{\Sigma}\boldsymbol{\Phi}^{T}|^{-1/2}\nonumber \\
 & = & (\sigma^{2})^{-D/2}|\boldsymbol{I}_{d}+\sigma^{-2}\boldsymbol{\Sigma}^{1/2}\boldsymbol{\Phi}^{T}\boldsymbol{\Phi}\boldsymbol{\Sigma}^{1/2}|^{-1/2}\nonumber \\
 & = & (\sigma^{2})^{-D/2}|\boldsymbol{I}_{d}+\sigma^{-2}\boldsymbol{\Sigma}|\nonumber \nonumber\\
 & = & (\sigma^{2})^{-D/2}\prod_{m=1}^{d}(\frac{1}{1+\sigma^{-2}\alpha_{m}^{2}})^{1/2} \nonumber
\end{eqnarray}

\end{proof}

\begin{thm}
\label{ParaExpanupperboundthm} Assume $\boldsymbol{\Omega}_{s,h}=\boldsymbol{\Psi}\boldsymbol{\Sigma}_{s,h}\boldsymbol{\Psi}^{T}+\sigma_{s}^{2}\boldsymbol{I}$
 where $\boldsymbol{\Psi}$ is a orthonormal $D \times D$ matrix and $\boldsymbol{\Sigma}_{s,h}$ is a $D \times D$ positive diagonal matrix. The distributions of $\boldsymbol{\Sigma}_{s,h}$ and $\sigma_{s}^{2}$ are defined in (6) and (7) in the submitted paper. Let $\boldsymbol{\Psi}^d$ denote the first $d$ columns of $\boldsymbol{\Psi}$, $\boldsymbol{\Sigma}_{s,h}^{d}=\diag{\alpha_{s,h,1}^2,\dots,\alpha_{s,h,d}^2}$ and let $\boldsymbol{\Omega}^{d}_{s,h}=\boldsymbol{\Psi}^d\boldsymbol{\Sigma}^{d}_{s,h}(\boldsymbol{\Psi}^d)^{T}+\sigma_{s}^{2}\boldsymbol{I}$. Then
for any $\epsilon>0$,
\[
Pr\{d_{\infty}(\boldsymbol{\Omega}_{s,h},\boldsymbol{\Omega}^{d}_{s,h})>\epsilon\}<\frac{6ba^{d}}{\epsilon(1-a)}
\]
for $d>2\log\{b/\epsilon(1-a)\}/\log(1/a)$, where $d_{\infty}(\boldsymbol{\Omega}_{s,h},\boldsymbol{\Omega}^{d}_{s,h})$ is defined as $\|\boldsymbol{\Omega}_{s,h}-\boldsymbol{\Omega}^{d}_{s,h}\|_{\infty}$. $\|A\|_{\infty}$ calculates the maximum absolute row sum of the matrix $A$, $b=E(\sigma_{s}^{2})$ and a = $E(\frac{1}{\tau_{s,h,1}})$.
\end{thm}
\begin{figure}
\begin{centering}
\includegraphics[width=0.70\textwidth]{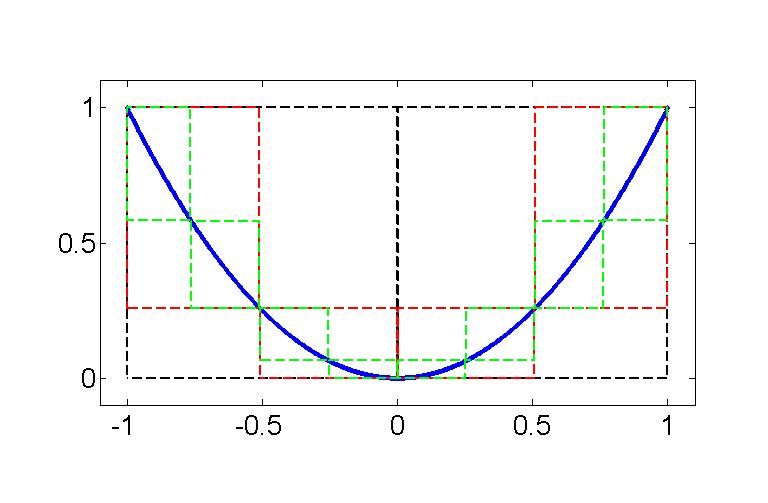}
\par\end{centering}
\caption{ A 4 level binary tree decomposition of a parabola using METIS, with the black rectangular denoting the second level cells, the red denoting the third level cells and the green denoting the leaf cells.}
\label{fig:TreeStructureDescription}
\end{figure}
\begin{proof}
With a slight abuse of notation, we write $u_{s,h,k}$ as $u$ and let $A=\prod_{m=1}^{K}\tau_{s,h,m}$. Let $\boldsymbol{\triangle}_{d}=\boldsymbol{\Psi}\boldsymbol{\Sigma}_{s,h}\boldsymbol{\Psi}^{T}-\boldsymbol{\Psi}^d\boldsymbol{\Sigma}^{d}_{s,h}(\boldsymbol{\Psi}^d)^{T}$, $\boldsymbol{\triangle}_{d}=\{a_{i.j}\}$ and $\boldsymbol{\Psi}=\{\psi_{i,j}\}$.
Clearly, $d_{\infty}(\boldsymbol{\Omega}_{s,h},\boldsymbol{\Omega}^{d}_{s,h})=\max_{1\leq i,j\leq D}|a_{i,j}^{d}|$,
and $a_{i,j}^{d}=\sum_{k=d+1}^{D}\alpha_{k}^{2}\psi_{i,k}\psi_{j,k}$.
By Cauchy-Schwartz inequality, 
\[
|\sum_{k=d+1}^{D}\alpha_{k}^{2}\psi_{i,k}\psi_{j,k}|\leq\max_{1\leq m\leq D}(\sum_{k=H+1}^{D}\alpha_{k}^{2}\psi_{m,k}^{2}).
\]
Since $\boldsymbol{\Psi}$ is orthonormal, we have $\psi_{i,j}^{2}\leq1$
for any $i$ and $j$. Hence
\[
d_{\infty}(\boldsymbol{\Omega}_{s,h},\boldsymbol{\Omega}^{d}_{s,h})\leq\sum_{k=d+1}^{D}\alpha_{k}^{2}.
\]
For a fixed $\epsilon>0$, by Chebyshev's inequalities
\begin{eqnarray*}
p\{d_{\infty}(\boldsymbol{\Omega}_{s,h},\boldsymbol{\Omega}^{d}_{s,h})\leq\epsilon\} & \geq & p\bigg\{\sum_{k=d+1}^{D}\alpha_{k}^{2}\leq\epsilon\bigg\}\\
 & = & E\bigg\{ p(\sum_{k=d+1}^{D}\alpha_{k}^{2}\leq\epsilon|\tau)\bigg\}\\
 & = & 1-E\bigg\{ p(\sum_{k=d+1}^{D}\alpha_{k}^{2}>\epsilon|\tau)\bigg\}\\
 & \geq & 1-E\bigg\{\frac{E(\sum_{k=d+1}^{D}\alpha_{k}^{2}|\tau)}{\epsilon}\bigg\}.
\end{eqnarray*}
By design we have $u\sim\mbox{Ga}_{(0,1)}(A+1,1)$
and $u$ and $\sigma_{s}^{2}$ are conditionally independent,
hence
\begin{eqnarray*}
E[(\frac{1}{u}-1)\sigma_{s}^{2}|\tau] & = & E[(\frac{1}{u}-1)|\tau]E(\sigma_{s}^{2}).
\end{eqnarray*}
Then we have
\begin{eqnarray*}
E[(\frac{1}{u}-1)|\tau] & = & \frac{\int_{0}^{1}(1/u-1)\frac{u^{A}}{\Gamma(A+1)}e^{-u}\mbox{d}u}{\int_{0}^{1}\frac{u^{A}}{\Gamma(A+1)}e^{-u}\mbox{d}u} = \frac{\int_{0}^{1}1/u\times u^{A}e^{-u}\mbox{d}u}{\int_{0}^{1}u^{A}e^{-u}\mbox{d}u}-1\\
 & = & \frac{\int_{0}^{1}u^{A-1}e^{-u}\mbox{d}u}{\int_{0}^{1}u^{A}e^{-u}\mbox{d}u}-1 = \frac{\frac{1}{A}u^{A}e^{-u}|_{0}^{1}+\int_{0}^{1}\frac{1}{A}u^{A}e^{-u}\mbox{d}u}{\int_{0}^{1}u^{A}e^{-u}\mbox{d}u}-1\\
 & = & \frac{e^{-1}}{A\int_{0}^{1}u^{A}e^{-u}\mbox{d}u}-1+\frac{1}{A}.
\end{eqnarray*}
Let $\gamma(s,x)=\int_{0}^{x}t^{s-1}e^{-t}\mbox{d}t$ be the lower incomplete Gamma function. Note that,
\begin{eqnarray*}
A\gamma(A+1,1) & = & \frac{A}{A+1}u^{A+1}e^{-u}|_{0}^{1}+\frac{A}{A+1}\gamma(A+2,1)\\
 & = & \frac{A}{A+1}e^{-1}+\frac{A}{A+1}\bigg[\frac{1}{A+2}e^{-1}+\frac{1}{A+2}\gamma(A+3,1)\bigg]\\
 & = & \lim_{K\to\infty}\bigg\{\sum_{k=1}^{K}\frac{\Gamma(A+1)^{2}}{\Gamma(A)\Gamma(A+k+1)}e^{-1}+A\Gamma(A+1)F(1;A+K,1)\bigg\}\\
 & = & \sum_{k=1}^{\infty}\frac{\Gamma(A+1)^{2}}{\Gamma(A)\Gamma(A+k+1)}e^{-1}\\
 & = & \sum_{k=1}^{\infty}\frac{A}{(A+1)(A+2)\dots(A+k)}e^{-1}
\end{eqnarray*}
where $F(x;a,b)$ is the cdf of $\mbox{Ga}(a,b)$ and $\lim_{a=\infty}F(1;a,1)=0$.
Furthermore we have 
\begin{eqnarray*}
\sum_{k=1}^{\infty}\frac{\Gamma(A+1)^{2}}{\Gamma(A)\Gamma(A+k+1)} & = \sum_{k=1}^{\infty}\frac{A}{(A+1)(A+2)\dots(A+k)} & \geq 1/2,
\end{eqnarray*}
and
\begin{eqnarray*}
1-\sum_{k=1}^{\infty}\frac{\Gamma(A+1)^{2}}{\Gamma(A)\Gamma(A+k+1)} & \leq 1-\frac{A}{A+1} &\leq \frac{1}{A},
\end{eqnarray*}
thus we have 
\begin{eqnarray*}
\frac{e^{-1}}{A\int_{0}^{1}u_{s,h,k}^{A}e^{-u_{h}}\mbox{d}u_{s,h,k}}-1+\frac{1}{A} & = & \frac{1}{\sum_{k=1}^{\infty}\frac{\Gamma(A+1)^{2}}{\Gamma(A)\Gamma(A+k+1)}}-1+\frac{1}{A}\\
 & = & \frac{1-\sum_{k=1}^{\infty}\frac{\Gamma(A+1)^{2}}{\Gamma(A)\Gamma(A+k+1)}}{\sum_{k=1}^{\infty}\frac{\Gamma(A+1)^{2}}{\Gamma(A)\Gamma(A+k+1)}}+\frac{1}{A}\\
 & \leq & \frac{1/A}{1/2}+\frac{1}{A}\\
 & = & \frac{3}{A}.
\end{eqnarray*}
Hence $E[(\frac{1}{u}-1)|\tau]\leq3/(\prod_{m=1}^{k}\tau_{s,h,m})$.
Based on this inequality, we have
\begin{eqnarray*}
\sum_{k=d+1}^{D}E\bigg\{ E[(\frac{1}{u}-1)\sigma_{s}^{2}|\tau]\bigg\} & \leq \sum_{k=d+1}^{D}E\bigg(\frac{3}{\prod_{m=1}^{k}\tau_{s,h,m}}\bigg)E(\sigma_{s}^{2})& \\
 & = \sum_{k=d+1}^{D}3ba^{k} \leq \frac{3ba^{d}}{1-a}&
\end{eqnarray*}
where $b=E(\sigma_{s}^{2})$ and a = $E(\frac{1}{\tau_{s,h,1}})$. Note that
$\tau_{s,h,m}\sim\mbox{Exp}_{[1,\infty)}(\lambda)$, thus $a<1$. By Fubini's
theorem, $E\bigg\{ E(\sum_{k=H+1}^{\infty}\alpha_{k}^{2}|\tau)\bigg\}=\sum_{k=d+1}^{\infty}E\bigg\{ E[(\frac{1}{u_{s,h,k}}-1)\sigma_{s}^{2}|\tau]\bigg\}$. Now use inequality $(1-x/2)>\exp (-x)$ if $0<x\leq1.5$
to get
\[
p\{d_{\infty}(\boldsymbol{\Omega}_{s,h},\boldsymbol{\Omega}^{d}_{s,h})\leq\epsilon\}\geq\exp\{\frac{-6ba^{d}}{\epsilon(1-a)}\}
\]
if $d>2\log\{b/\epsilon(1-a)\}/\log(1/a)$. Hence,
\[
p\{d_{\infty}(\boldsymbol{\Omega}_{s,h},\boldsymbol{\Omega}^{d}_{s,h})>\epsilon\}\leq1-\exp\{\frac{-6ba^{d}}{\epsilon(1-a)}\}\leq\frac{6ba^{d}}{\epsilon(1-a)},
\]
since $6ba^{d}/\{\epsilon(1-a)\}<1$.\end{proof}

\begin{thm}
\label{ScaleTruncationBound}Let
\[
f^{L}(\boldsymbol{y}_{i})=\sum_{s=1}^{L}\sum_{h=1}^{2^{s}}\tilde{\pi}_{s,h}\mathcal{N}_{D}(\boldsymbol{y}_{i};\boldsymbol{\mu}_{s,h},\boldsymbol{\Phi}_{s,h}\boldsymbol{\Sigma}_{s,h}\boldsymbol{\Phi}_{s,h}^{T}+\sigma_{s}^{2}\boldsymbol{I})
\]
denote the approximation at scale $L$, let $P(B)=\int_{B}f(\boldsymbol{y}_{i})dy$
and $P^{L}(B)=\int_{B}f^{L}(\boldsymbol{y}_{i})dy$, for all $B \subset \Re^D$ denote
the probability measures corresponding to density $f(\boldsymbol{y}_{i})$ and $f^{L}(\boldsymbol{y}_{i})$.
Then we have,
\[
d_{TV}(P_{L},P)<\bigg(\frac{a_{S}}{1+a_{S}}\bigg)^{L},
\]
where $d_{TV}(P_{L},P)$ denotes the total variation distance between $P_{L}(B)$ and $P(B)$.
\end{thm}
\begin{proof}
The total variation distance is given by
\begin{eqnarray*}
d_{TV}(P_{L},P) & = & \sup_{B\in\Re^{D}}|P^{L}(B)-P(B)|\\
 & = & \sup_{B\in\Re^{D}}|\sum_{h=1}^{2^{L}}\tilde{\pi}_{s,h}N(B;\boldsymbol{\mu}_{s,h},\boldsymbol{\Phi}_{s,h}\boldsymbol{\Sigma}_{s,h}\boldsymbol{\Phi}_{s,h}^{T}+\sigma_{s}^{2}\boldsymbol{I})-...\\
 &  & \sum_{s=L}^{\infty}\sum_{h=1}^{2^{s}}\pi_{s,h}N(B;\boldsymbol{\mu}_{s,h},\boldsymbol{\Phi}_{s,h}\boldsymbol{\Sigma}_{s,h}\boldsymbol{\Phi}_{s,h}^{T}+\sigma_{s}^{2}\boldsymbol{I})|\\
 & \leq & \max\{\sum_{h=1}^{2^{L}}\tilde{\pi}_{s,h},\sum_{s=L}^{\infty}\sum_{h=1}^{2^{s}}\pi_{s,h}\}\\
 & = & \max\bigg\{2^{L}(\frac{a_{S}}{1+a_{S}})^{L-1}\frac{1}{1+a_{S}}2^{-L},\sum_{s=L}^{\infty}2^{s}\frac{1}{1+a_{S}}(\frac{a_{S}}{2+2a_{S}})^{s}\bigg\}\\
 & = & \sum_{s=L}^{\infty}\frac{1}{1+a_{S}}(\frac{a_{S}}{1+a_{S}})^{s}\\
 & = & (\frac{a_{S}}{1+a_{S}})^{L}
\end{eqnarray*}
\end{proof}

\section{Posterior Conditional Derivation}

Based on the likelihood function (\ref{eq:Likelihood}), the derivation of conditional posterior of $\sigma_{s}^{-2}$ is given
by
\begin{eqnarray*}
p(\sigma_{s}^{-2}|-) & \propto & (\sigma_{s}^{-2})^{a_{\sigma}-1}\exp(-b_{\sigma}\sigma_{s}^{-2})\prod_{y_{i}\in C_{s}}(\sigma_{s}^{2})^{-D/2}\\
 &  & \exp\bigg\{-\frac{1}{2}\sigma_{s}^{-2}(A_{s,h,i}-\sum_{j=1}^{d}(1-u_{s,h,j})(Z_{s,h,i}^{(j)})^{2})\bigg\}\\
 & \propto & (\sigma_{s}^{-2})^{Dn_{s}/2+a_{\sigma}-1}\\
 &  & \exp\bigg\{-\sigma_{s}^{-2}[\frac{1}{2}\sum_{y_{i}\in C_{s}}(A_{s,h,i}-\sum_{j=1}^{d}(1-u_{s,h,j})(Z_{s,h,i}^{(j)})^{2})+b_{\sigma}]\bigg\}.
\end{eqnarray*}
The derivation of conditional posterior of $u_{s,h,m}$ is given by
\begin{eqnarray*}
p(u_{s,h,m}|-) & \propto & \prod_{y_{i}\in C_{s,h}}u_{s,h,m}^{1/2}\exp\bigg\{-\frac{1}{2}\sigma_{s}^{-2}u_{s,h,m}(Z_{s,h,i}^{(m)})^{2}\bigg\}\\
 &  & u_{s,h,m}^{\prod_{j=1}^{m}\tau_{s,h,j}-1}\exp\{-u_{s,h,m}\}I_{(0,1)}\\
 & \propto & u_{m,s,h}^{\prod_{j=1}^{m}\tau_{s,h,j}+n_{s,h}/2-1}\\
 &  & \exp\bigg\{-[1+\frac{1}{2}\sigma_{s}^{-2}\sum_{y_{i}\in C_{s,h}}(Z_{s,h,i}^{(m)})^{2}]u_{s,h,m}\bigg\} I_{(0,1)}.
\end{eqnarray*}
The derivation of conditional posterior of $\tau_{s,h,m}$ is given by
\begin{eqnarray*}
p(\tau_{s,h,m}|-) & \propto & (\prod_{j>m-1}u_{j,s,h})^{\tau_{s,h,j}}\exp\{-a_{\tau}\tau_{s,h,m}\}I_{[1,\infty)}\\
 & \propto & \exp\bigg\{-[a_{\tau}-ln(\prod_{j>m-1}u_{s,h,j})]\tau_{s,h,m}\bigg\}
\end{eqnarray*}

\section{Missing Data Imputation}

\begin{prop}
\label{Prop MD}
For node $(s,h)$, introduce augmented data $\boldsymbol{\eta}_{i}$ such that $(\boldsymbol{y}_{i}|\boldsymbol{\eta}_{i},\boldsymbol{\Theta},s_{i}=s,h_{i}=h) \sim \mathcal{N}_{D}(\boldsymbol{\mu}_{s,h}+\boldsymbol{\Phi}_{s,h}\boldsymbol{\eta}_{i},\sigma_{s}^{2}\boldsymbol{\mbox{I}}_{D})$ and $(\boldsymbol{\eta}_{i}|\boldsymbol{\Theta},s_{i}=s,h_{i}=h) \sim \mathcal{N}_{d}(0,\boldsymbol{\Sigma}_{s,h})$, for $i=1,\dots,n$. Then
we have the conditional distribution with $\boldsymbol{\eta}_{i}$ marginalized out equal $(\boldsymbol{y}_{i}|\boldsymbol{\Theta},s_{i}=s,h_{i}=h) \sim \mathcal{N}_{D}(\boldsymbol{\mu}_{s,h}+\boldsymbol{\Phi}_{s,h}\boldsymbol{\Sigma}_{s,h}\boldsymbol{\Phi}_{s,h}^{T},\sigma_{s}^{2}\boldsymbol{\mbox{I}}_{D})$. Furthermore, conditional on $s_{i}=s$ and $h_{i}=h$ we have
\begin{eqnarray*}
\boldsymbol{\eta}_{i}|\boldsymbol{y}_{O},\boldsymbol{\Theta} \sim \mathcal{N}_{d}(\hat{\boldsymbol{\mu}}_{\eta},\hat{\boldsymbol{\Sigma}}_{\eta}), & \mbox{ }& \boldsymbol{y}_{M}|\boldsymbol{\eta}_{i},\boldsymbol{y}_{O},\boldsymbol{\Theta} \sim \mathcal{N}_{m_{i}}(\boldsymbol{\mu}_{M}+\boldsymbol{\Phi}_{M}\boldsymbol{\eta}_{i},\sigma_{s}^{2}\mbox{I}_{m_{i}}),
\end{eqnarray*}
where $\hat{\boldsymbol{\Sigma}}_{\eta}=\big(\boldsymbol{\Sigma}_{s,h}\boldsymbol{\Phi}_{O}^{T}\boldsymbol{\Phi}_{O}/\sigma_{s}^{2}+\mbox{I}\big)^{-1}\boldsymbol{\Sigma}_{s,h}$
and $\hat{\boldsymbol{\mu}}_{\eta}=\hat{\boldsymbol{\Sigma}}_{\eta}\boldsymbol{\Phi}_{O}^{T}(\boldsymbol{y}_{O}-\boldsymbol{\mu}_{O})/\sigma_{s}^{2}$,
\end{prop}
\begin{proof}
The proposition can be easily proved using Bayes rule. The joint density
of $(\boldsymbol{y}_{O},\boldsymbol{y}_{M},\boldsymbol{\eta}_{i}|\boldsymbol{\Theta})$
is given by
\begin{eqnarray*}
p(\boldsymbol{y}_{O},\boldsymbol{y}_{M},\boldsymbol{\eta}_{i}|\boldsymbol{\Theta},s_{i}=s,h_{i}=h) & \propto & \exp\bigg\{-\frac{\|\boldsymbol{y}_{i}-\boldsymbol{\Phi}\boldsymbol{\eta}_{i}-\boldsymbol{\mu}\|_{2}}{2\sigma_{s}^{2}}-\frac{\boldsymbol{\eta}_{i}^{T}\boldsymbol{\Sigma}_{s,h}^{-1}\boldsymbol{\eta}_{i}}{2}\bigg\}\\
 & \propto & \exp\bigg\{-\frac{\|\boldsymbol{y}_{M}-\boldsymbol{\Phi}_{M}\boldsymbol{\eta}_{i}-\boldsymbol{\mu}_{M}\|_{2}}{2\sigma_{s}^{2}}\\
 &  & -\frac{\|\boldsymbol{y}_{O}-\boldsymbol{\Phi}_{O}\boldsymbol{\eta}_{i}-\boldsymbol{\mu}_{O}\|_{2}}{2\sigma_{s}^{2}}-\frac{\boldsymbol{\eta}_{i}^{T}\boldsymbol{\Sigma}_{s,h}^{-1}\boldsymbol{\eta}_{i}}{2}\bigg\}.
\end{eqnarray*}
Hence the conditional density $(\boldsymbol{y}_{M}|\boldsymbol{\eta}_{i},\boldsymbol{y}_{O},\boldsymbol{\Theta},s_{i}=s,h_{i}=h)$
is given by
\begin{eqnarray*}
p(\boldsymbol{y}_{M}|\boldsymbol{\eta}_{i},\boldsymbol{y}_{O},\boldsymbol{\Theta},s_{i}=s,h_{i}=h) & \propto & \exp\bigg\{-\frac{\|\boldsymbol{y}_{M}-\boldsymbol{\Phi}_{M}\boldsymbol{\eta}_{i}-\boldsymbol{\mu}_{M}\|_{2}}{2\sigma_{s}^{2}}\bigg\}.
\end{eqnarray*}
The marginal conditional density $(\boldsymbol{\eta}_{i}|\boldsymbol{y}_{O},\boldsymbol{\Theta},s_{i}=s,h_{i}=h)$
is given by 
\begin{eqnarray*}
\mbox{p}(\boldsymbol{\eta}_{i}|\boldsymbol{y}_{i}^{O},\boldsymbol{\Theta},s_{i}=s,h_{i}=h) & \propto & \int\mbox{p}(\boldsymbol{y}_{M},\boldsymbol{\eta}_{i}|\boldsymbol{y}_{O})\mbox{d}\boldsymbol{y}_{M}\\
 & \propto & \exp\bigg\{\frac{\|\boldsymbol{y}_{O}-\boldsymbol{\Phi}_{O}\boldsymbol{\eta}_{i}-\boldsymbol{\mu}_{O}\|_{2}}{2\sigma_{s}^{2}}-\frac{\boldsymbol{\eta}_{i}^{T}\boldsymbol{\Sigma}_{s,h}^{-1}\boldsymbol{\eta}_{i}}{2}\bigg\}.
\end{eqnarray*}
\end{proof}

To finish the missing data imputation algorithm, the conditional posterior distribution of the membership variable $(s_{i},h_{i})$ of partially observed subject $i$, $p(s_{i},h_{i}|\boldsymbol{y}_{O},\boldsymbol{\Theta})$ is needed. $\boldsymbol{y}_{M}$ has been marginalized out to reduce the sample autocorrelation, and the distribution is given by
\begin{eqnarray}
\mbox{p}(s_{i},h_{i}|\boldsymbol{y}_{O},\boldsymbol{\Theta}) & \propto & \int p(\boldsymbol{y}_{M},\boldsymbol{y}_{O},\boldsymbol{\Theta},s_{i},h_{i})\mbox{d}\boldsymbol{y}_{M} \nonumber\\
 & \propto & \int\pi_{s_{i},h_{i}}\mathcal{N}_{D}(\boldsymbol{y}_{i};\boldsymbol{\mu}_{s_{i},h_{i}},\boldsymbol{\Phi}_{s_{i},h_{i}}\boldsymbol{\Sigma}_{s_{i},h_{i}}\boldsymbol{\Phi}_{s_{i},h_{i}}^{T}+\sigma_{s_{i}}^{2}\boldsymbol{\mbox{I}})\mbox{d}\boldsymbol{y}_{M}.\nonumber
\end{eqnarray}
With a slight abuse of notation, we write $\boldsymbol{\Phi}$ as
$\boldsymbol{\Phi}_{s_{i},h_{i}}$, $\boldsymbol{\Sigma}$ as $\boldsymbol{\Sigma}_{s_{i},h_{i}}$,
$\sigma^{2}$ denote $\sigma_{s_{i}}^{2}$ and $\boldsymbol{\mu}$
as $\boldsymbol{\mu}_{s_{i},h_{i}}$. By properties of multicariate
Gaussian, we have
\begin{equation*}
\int\mathcal{N}_{D}(\boldsymbol{y}_{i};\boldsymbol{\mu},\boldsymbol{\Phi}\boldsymbol{\Sigma}\boldsymbol{\Phi}^{T}+\sigma^{2}\mbox{I})\mbox{d}\boldsymbol{y}_{M}=\mathcal{N}_{D-m_{i}}(\boldsymbol{y}_{O};\boldsymbol{\mu}_{O},\boldsymbol{\Phi}_{O}\boldsymbol{\Sigma}\boldsymbol{\Phi}_{O}^{T}+\sigma^{2}\boldsymbol{I}).
\end{equation*}
Hence we have $p(s_{i},h_{i}|\boldsymbol{y}_{O},\boldsymbol{\Theta})\propto\mathcal{N}_{D-m_{i}}(\boldsymbol{y}_{O};\boldsymbol{\mu}_{O},\boldsymbol{\Phi}_{O}\boldsymbol{\Sigma}\boldsymbol{\Phi}_{O}^{T}+\sigma^{2}\boldsymbol{I})$. Directly computing this value includes inverting a $(D-m_{i})\times(D-m_{i})$ matrix, which
is computational intractable when $D-m_{i}$ is large. With basic linear algebra, we have
\begin{equation*}
\big|\boldsymbol{\Phi}_{O}\boldsymbol{\Sigma}\boldsymbol{\Phi}_{O}^{T}+\sigma^{2}\boldsymbol{I}\big|=(\sigma^{2})^{D-m_{i}}\big|\boldsymbol{I}+\boldsymbol{\Sigma}\boldsymbol{\Phi}_{O}^{T}\boldsymbol{\Phi}_{O}/\sigma^{2}\big|,
\end{equation*}
\begin{equation*}
\big(\boldsymbol{\Phi}_{O}\boldsymbol{\Sigma}\boldsymbol{\Phi}_{O}^{T}+\sigma^{2}\boldsymbol{I}\big)^{-1}=\frac{\boldsymbol{I}}{\sigma^{2}}-\frac{\boldsymbol{\Phi}_{O}\big(\boldsymbol{I}+\boldsymbol{\Sigma}\boldsymbol{\Phi}^{T}\boldsymbol{\Phi}_{O}/\sigma^{2}\big)^{-1}\boldsymbol{\Sigma}\boldsymbol{\Phi}_{O}^{T}}{\sigma^{4}}.
\end{equation*}
Hence we have
\begin{equation}
\tag{A2}
\label{eq:ClPr} 
\begin{aligned}
& \mathcal{N}_{D-m_{i}}(\boldsymbol{y}_{O};\boldsymbol{\mu}_{O},\boldsymbol{\Phi}_{O}\boldsymbol{\Sigma}\boldsymbol{\Phi}_{O}^{T}+\sigma^{2}\boldsymbol{I}) \\
= & (2\pi\sigma^{2})^{-(D-m_{i})/2}\big|\boldsymbol{I}+\boldsymbol{\Sigma} \boldsymbol{A}/\sigma^{2}\big|^{-1/2}\\
& \times\exp\bigg\{-\frac{B_{i}}{2\sigma^{2}}+\frac{\boldsymbol{C}_{i}^{T}\big(\boldsymbol{\Sigma}^{-1}+\boldsymbol{A}/\sigma^{2}\big)^{-1}\boldsymbol{C}_{i}}{2\sigma^{4}}\bigg\}
\end{aligned}
\end{equation}

where $\boldsymbol{A}=\boldsymbol{\Phi}_{O}^{T}\boldsymbol{\Phi}_{O}$, $B_{i}=\|\boldsymbol{y}_{O}-\boldsymbol{\mu}_{O}\|_{2}$
and $\boldsymbol{C}_{i}=\boldsymbol{\Phi}_{O}^{T}(\boldsymbol{y}_{O}-\boldsymbol{\mu}_{O})$.
Note that $\boldsymbol{A}$, $B_{i}$ and $\boldsymbol{C}_{i}$ can be computed
before the MCMC algorithm with a computational cost being $O\big((D-m_{i}) d\big)$. Within the MCMC, the cost to compute (\ref{eq:ClPr})
is only $O(d^3)$.

\section{Simulation Studies}

In the missing data imputation simulatoin study, we simulated 100 independent samples of size $n=600$ from different scenarios as follows.

\begin{enumerate}
\item[] \textbf{Scenario 1-6:} Data $\boldsymbol{y}_{i}$, for $i=1,\dots,600$, were generated from $\mathcal{N}_{D}(0,\boldsymbol{\Lambda}\boldsymbol{\Lambda}^{T}+\sigma^2\boldsymbol{I})$. $\boldsymbol{\Lambda}$ is a $D\times p$ matrix with each entry generated from $\mathcal{N}(0,25)$ and $10\sigma^2$ was generated from $\chi_{(1)}$. This scenario includes different cases where $p\in\{10,50\}$, $D\in\{5000,10000,15000\}$ and with or without a 20\% missing data. We fixed the upper bound to $d = 100$.

\item[] \textbf{Scenario 7-9:} 3\textendash{}D data $\boldsymbol{\eta}_{i}$, for $i=1,\dots,600$, were generated on the Swissroll with Gaussian noise distributed as $\mathcal{N}(0,2.5\times10^{-5})$ along each dimension. Data $\boldsymbol{y}_{i}$, for $i=1,\dots,600$, were obtained by $\boldsymbol{y}_{i}=\boldsymbol{\Lambda}\boldsymbol{\eta}_{i}$ where $\boldsymbol{\Lambda}$ were generated in the same way as in Scenario 1. This scenario includes different cases where $D\in\{5000,10000,15000\}$ and with or without a 20\% missing data. We fixed the upper bound to $d = 10$.
\end{enumerate}

The average inclusion probabilities of each presetted dimensions were computed in the following way. Let $\mathcal{R}_{s,h}^{t}$ denotes the set of retained column indices of node $(s,h)$ at the $t$th iteration, and let $(s_{i}^{t},h_{i}^{t})$ denote the node index of the $i$th observation at the $t$th iteration. Then the inclusion probability of dimension $j=1,2,\dots,10$ in scenario 2 is given by
\begin{equation*}
p_{j}^{inclu}=\frac{1}{n_{adapt}\times N}\sum_{t\mbox{: }adapt}\sum_{i=1}^{N}I_{(j\in\mathcal{R}_{s_{i}^{t},h_{i}^{t}}^{t})}
\end{equation*}
where $n_{adapt}$ denotes the number of adaptation steps during the MCMC collection interval.

\end{document}